\documentclass[twocolumn]{IEEEtran}

\usepackage[utf8]{inputenc}
\usepackage[ruled,vlined]{algorithm2e}
\usepackage{amsmath}
\usepackage{amsthm}
\usepackage{amsfonts}
\usepackage{subfigure}
\usepackage{mathrsfs}
\usepackage{pifont}
\usepackage{amssymb}
\usepackage{verbatim}
\usepackage{upgreek}
\usepackage{color}
\usepackage{epsfig}
\usepackage{booktabs}
\usepackage{bm}
\usepackage{setspace}
\usepackage[hyphens]{url}
\usepackage{graphicx}
\usepackage[overload]{empheq}
\usepackage{fancyhdr}
\usepackage[colorlinks]{hyperref}
\usepackage{cite}
\usepackage{float}

\newcommand{\bzero}{\mbox{\boldmath{$0$}}}
\newcommand{\bA}{\mbox{\boldmath{$A$}}}
\newcommand{\ba}{\mbox{\boldmath{$a$}}}
\newcommand{\bb}{\mbox{\boldmath{$b$}}}
\newcommand{\bB}{\mbox{\boldmath{$B$}}}
\newcommand{\bc}{\mbox{\boldmath{$c$}}}
\newcommand{\bC}{\mbox{\boldmath{$C$}}}

\newcommand{\bg}{\mbox{\boldmath{$g$}}}

\newcommand{\bH}{\mbox{\boldmath{$H$}}}

\newcommand{\bI}{\mbox{\boldmath{$I$}}}

\newcommand{\bp}{\mbox{\boldmath{$p$}}}
\newcommand{\bq}{\mbox{\boldmath{$q$}}}

\newcommand{\bU}{\mbox{\boldmath{$U$}}}

\newcommand{\bX}{\mbox{\boldmath{$X$}}}
\newcommand{\bx}{\mbox{\boldmath{$x$}}}
\newcommand{\by}{\mbox{\boldmath{$y$}}}
\newcommand{\bz}{\mbox{\boldmath{$z$}}}

\newcommand{\diag}{\mbox{\boldmath\bf diag}\, }

\newtheorem{theorem}{Theorem}[section]

\newtheorem{lemma}[theorem]{Lemma}
\newtheorem{proposition}[theorem]{Proposition}

\title{Enhanced Target Localization with Deployable Multiplatform Radar Nodes Based on Non-Convex Constrained Least Squares Optimization}

\author{Augusto Aubry, \emph{Senior Member, IEEE}, Paolo Braca, \emph{Senior Member, IEEE}, Antonio De Maio, \emph{Fellow, IEEE}, and Angela Marino, \emph{Student Member, IEEE}

 \thanks{A.~Aubry, A.~De~Maio (corresponding author), and A. Marino are with the Department of Electrical and Information Technology Engineering, University of Naples Federico II, $80125$ Napoli, Italy.
 Email: augusto.aubry@unina.it, ademaio@unina.it, angela.marino@unina.it.}
 \thanks{P.~Braca is with Centre for Maritime Research and Experimentation (CMRE).
Email: paolo.braca@cmre.nato.int
}
 }

\begin{document}

\maketitle

\begin{abstract}
A new algorithm for 3D localization in multiplatform radar networks, comprising one transmitter and multiple receivers, is proposed. To take advantage of the monostatic sensor radiation pattern features, ad-hoc constraints are imposed in the target localization process. Therefore, the localization problem is formulated as a non-convex constrained Least Squares (LS) optimization problem which is globally solved in a quasi-closed-form leveraging Karush-Kuhn-Tucker (KKT) conditions. The performance of the new algorithm is assessed in terms of Root Mean Square Error (RMSE) in comparison with the benchmark Cramer Rao Lower Bound (CRLB) and some competitors from the open literature. The results corroborate the effectiveness of the new strategy which is capable of ensuring a lower RMSE than the counterpart methodologies especially in the low Signal to Noise Ratio (SNR) regime.
\end{abstract}


\section{Introduction}\label{Introduction}
{{Multiplatform radar networks (MPRNs) exhibit  significant  advantages over the monostatic radar configuration and have received great attention in the recent past \cite{sensors_as_robots}.   Employing a constellation of multiple 
deployable platforms allows to enlarge the surveillance area, to improve data reliability and accuracy, to enhance the fault tolerance, and to improve the data utilization of the system. 
Furthermore, such multistatic configurations exploiting spatial diversity can grant improved target detectability, especially against  low-observable and stealth targets \cite{mimo_book,bistatic_book}.  Last but not least, MPRNs also endow a better resistance to electronic countermeasures, such as focused jamming \cite{multistatic_jamming}, and can reduce the deleterious effects of shadowing, obscuring large angular sections of the coverage region.
MPRNs can be classified in two main groups \cite{multiplatform_workshop}. The first group can be defined as a constellation of multiple Autonomous Radio Frequency (RF) sensor Nodes (A-RFN) working autonomously towards a specific task.  The second group embraces a set of Cooperating Radio Frequency sensor Nodes (C-RFN) which work together to achieve a common goal. The two groups exhibit different characteristics: A-RFNs are highly decentralized, do not demand (or require limited) communication among nodes, and are resilient to possible sensor failures. C-RFNs involve information sharing among the different nodes and the availability of a central sensor (or possibly multiple central sensors) performing the final processing. Needless to say, there is a strong dependence on the network infrastructure and, in particular on the synchronization protocol \cite{intelligent_radar_network}. In this respect, it is worth pointing out that many synchronization and geolocation issues that were previously critical are now affordable  using Global Positioning System (GPS) and highly stable GPS Disciplined Oscillators (GPSDOs) at each node of the network \cite{sync,sync_netted_radar,sync_oscillators}. Besides, techniques to grant synchronization in a GNSS-denied environment are also available \cite{sync_timing}.
As to the control of C-RFNs, which includes resource allocation and sensor management, it represents a major demand (especially for moving platforms) together with resource allocation and sensor management \cite{sensor_management,Distributed_MultiPlatform,quality_of_service}. Not surprisingly, C-RFNs are in general less robust than A-RFM even if a cooperative protocol can theoretically grant superior performance than the autonomous setup. For instance, the accuracy of target positioning is strongly influenced by the network geometrical configuration and even more by the baseline lengths and orientations. This geometric diversity can be deemed as the key ingredient for performance enhancement thanks to its inherent flexibility and the possibility to optimize dynamically the number and the locations of the individual platforms \cite{adaptive_radar_network}. Remarkably, it paves the way to new opportunities and challenges connected with the chance of using low cost receiving units (possibly expendable and heterogeneous) such as Unmanned Aerial Vehicles (UAV) deployable on the base of a specific task.}}

Performance gains in multiplatform systems due to waveform and frequency diversity are discussed in \cite{multiplatform_air_surveillance}. Therein
examples of western multiplatform systems are presented together with important standardization and interoperability issues.
In \cite{Adaptive_Waveform_Multistatic} and \cite{Adaptive_Waveform_Multistatic_ieee} the problem of adaptive waveform selection for target tracking by a multistatic radar system consisting of a dedicated transmitter and multiple receivers is considered. Prototypes of ground-based multistatic sensors with one transmitter and up to three receive nodes (a monostatic-pair plus two widely-spaced bistatic pairs) are
NetRAD \cite{NetRad,NETRAD_UAV} and its evolution NextRAD \cite{nextrad,nextrad2} capable of performing polarimetric acquisitions. NetRAD has been also successfully used to detect and track UAV with very low radial velocity \cite{NETRAD_UAV}.

Target localization with a multistatic radar is  addressed in \cite{passive_multistatic}, where two methods for calculating the Cartesian position are  presented resorting to Spherical-Interpolation
(SI) and Spherical-Intersection (SX).
A localization scheme exploiting both Time Of Arrival (TOA) from the transmitter to a specific receiver and Angle Of Arrival (AOA)   is
proposed in \cite{AOA_multistatic}, 
that applies the weighted least squares method to estimate the target location and shows that the RMSE decreases as the number of
multistatic radar receivers increases under the assumption of Gaussian measurement errors.
 
An improved method for moving target localization with a noncoherent Multiple-Input Multiple-Output (MIMO) radar system having widely separated antennas is proposed in \cite{localization_moving_target_MIMO_noncoherent}. Specifically, the proposed method is based on the Two-Stage Weighted Least Squares (2SWLS), and a closed-form solution is derived \cite{localization_moving_target_MIMO_noncoherent}. In \cite{localization_moving_target_MIMO}, for the same problem of moving target localization, the authors propose two methods, in which the parameters used are the joint of AOA, Frequency-Of-Arrival (FOA) and TOA.

\raggedbottom

This paper proposes a novel approach for 3D localization in multiplatform systems with one transmitter and multiple receivers.  At the  design phase, 
angular constraints are forced on the target position to capitalize the information embedded into the characteristics of the monostatic radiation pattern\footnote{Angular constraints, induced by the radiation pattern, have been already explored and proved effective in other localization systems for 2D scenarios \cite{acls,aacls,AACLS_journal}. }. Therefore, localization is formulated as a constrained Least Squares (LS) problem whose optimal solution provides the Cartesian coordinates of the target. The resulting non-convex optimization is efficiently handled invoking the Karush-Kuhn-Tucker (KKT) optimality conditions \cite{Bertsekas16}. Specifically, a quasi-closed-form
global optimal solution, i.e., depending only on elementary functions and roots of polynomial equations, is computed leveraging to an ad-hoc partition of the feasible set as well as the regularity of its points.  
The proposed localization method, referred as Angular and Range Constrained Estimator (ARCE), is tested in different illustrative examples, and compared with some counterparts, e.g., the Unconstrained TDOA Like (U-TDOA), the Range-Only Constrained Estimator, the Two-Step Estimation-$1$ (TSE-$1$) \cite{TSE}, extended to the 3D scenario, and Two-Step Estimation-$2$ (TSE-$2$) \cite{Multistatic_joint_localization}, specialized for known transmitter position.

The paper is organized as follows. Section \ref{system model} introduces the system model and defines the constraints bestowed by the monostatic system radiation characteristics. Section \ref{Problem Formulation} formulates the constrained LS estimation problem and presents the algorithm yielding a quasi-closed-form solution. Section \ref{analysis} deals with the performance analysis and comparisons. Section \ref{conclusions} draws some conclusions and highlights some possible future research avenues.

\subsection{Notation}

We adopt the notation of using boldface for vectors $\ba$ (lower case), and matrices $\bA$ (upper case). The $n$-th element of $\bm{a}$ and the $(m,n)$-th entry of $\bm{A}$ are denoted by $a_n$ and $\bm{A}_{m,n}$, respectively. The symbols $(\cdot)^T$  indicates the transpose  operator. 
$\lceil\cdot\rceil$ denotes the
operation of rounding up to the nearest integer.
$\bA^\dagger$ represents the Moore-Penrose inverse of the matrix $\bA$.
$\bI$ and ${\bf 0}$ denote respectively the identity matrix and the matrix with zero entries (their size is determined from the context). $\mathbf{1}_N$ and
$\mathbf{0}_N$ are $N$-length vectors of ones and zeros. ${\mathbb{R}}^N$, ${\mathbb{R}}^{N,M}$,  and ${\mathbb{S}}^N$ are respectively the sets of $N$-dimensional vectors of real numbers, of $N\times M$ real matrices, and of $N\times N$ symmetric matrices. $\diag(\ba)$ indicates the diagonal matrix whose $i$-th diagonal element is the $i$-th entry of $\ba$.
The 
symbol $\succeq$ (and its strict form $\succ$) is used to indicate generalized matrix inequality: for any $\bA\in{\mathbb{S}}^N$, $\bA\succeq\bzero$ means that $\bA$ is a positive semi-definite matrix ($\bA\succ\bzero$ for positive definiteness). $\lambda_1(\bX), \ldots, \lambda_N(\bX)$, with $\lambda_1(\bX)\geq \ldots \geq \lambda_N(\bX)$, denote the  eigenvalues of $\bX\in {\mathbb{S}}^N$, arranged in decreasing order.   
The Euclidean norm of the vector $\bx \in \mathbb{R}^N$ is denoted as $\|\bx\|$. 

\section{System Model}\label{system model}
Let us consider a multistatic radar network with an active sensor and $N$ receivers, as illustrated in Fig. \ref{fig:scenario} and denote by:

\begin{itemize}
\item $\bp=[x_p,y_p,z_p]^T\in\mathbb{R}^3$ the target position;
\item $\bp_{r_0}=[x_0,y_0,z_0]^T\in\mathbb{R}^3$ the active radar position (without loss of generality, it is assumed coinciding with the reference system origin, i.e., $\bp_{r_0}=[0,0,0]^T$);
\item ${\bp_r}_i = [x_{r_i},y_{r_i},z_{r_i}]^T\in\mathbb{R}^3$ the position of the $i$-th receiver, $i=1,\ldots,N$.

\end{itemize}
\begin{figure}[hb]
    \centering
    \includegraphics[width=0.8\columnwidth]{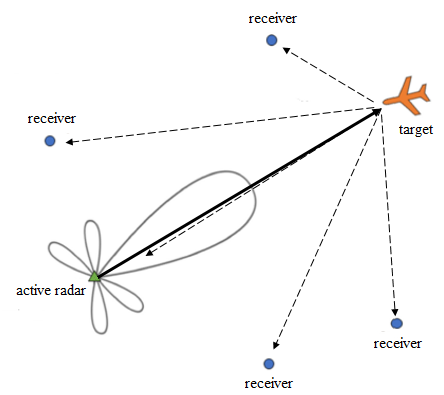}
        \caption{Pictorical representation of a surveillance system including a monostatic radar and $N=4$ receiver nodes.}
    \label{fig:scenario}
\end{figure}

Letting 
\begin{equation}\label{model_equation}
\begin{split}
  \tilde{\tau}_i= \frac{1}{c} \left(\|\bp\| + \| \bp - {\bp_r}_i \|\right),~ i=0,\ldots,N
\end{split}
\end{equation}
the noise-free output of the cross-correlation
based processing (with $c$ the speed of light) at the $i$-th bistatic (or monostatic, if $i=0$) pair, the following $N+1$ noisy delay measurements are available at the $i$-th receiver

\begin{equation}\label{bistatic_measurement}
  \tau_i= \tilde{\tau}_i + n_i, ~i=0,\ldots,N. 
\end{equation}
They are sent to the active radar that plays the role of a fusion node,  to determine an estimate of the target position.
In \eqref{bistatic_measurement}, $n_0,\ldots,n_N$ are statistically independent zero-mean
(usually Gaussian distributed) random variables with variance $\sigma^2_0,\ldots,\sigma^2_N$ given by

\begin{equation}\label{variance_bistatic_measurements}
    \sigma_i = \frac{1}{B \sqrt{2\mbox{SNR}_i} },~i =0,\ldots,N
\end{equation}
where $B$ represents the frequency bandwidth of the $i$-th receiver and $\mbox{SNR}_i$ denotes the signal to noise ratio (SNR) of the $i$-th bistatic pair (i.e., radar/$i$-th receiver) or, if $i=0$, of the monostatic radar, computed  via the bistatic and monostatic radar range equation \cite{G_trus_region,Richards}, respectively.
Now, elaborating on \eqref{model_equation}, it is possible to get an equivalent form which is fundamental for the development of the proposed estimation algorithm. To this end, let 
\begin{equation}\label{b_i}
b_i=c{\tilde\tau}_i - \frac{c{\tilde\tau}_0}{2}, \, \, \quad i=1,\ldots, N \, ,
\end{equation}
and (for $i=0$)
\begin{equation}\label{b_0}
b_0= \frac{c{\tilde\tau}_0}{2}. 
\end{equation}
Equation \eqref{model_equation} can be recast as:

\begin{equation}
\begin{split}
  & \|\bp\|^2 -2x_p x_{r_i} - 2y_p y_{r_i} +\\
&-2z_p z_{r_i} + \|\bp_{r_i}\|^2= b_i^2 \, \, \quad i=1,\ldots, N \, , 
\end{split}
\end{equation}
which is equivalent to
\begin{equation} \label{model_equation_2}
\begin{cases}
- 2x_p x_{r_i} - 2y_p y_{r_i} - 2z_p z_{r_i} - g_i = 0\\
g_i = b_i^2 - b_0^2 - x_{r_i}^2 - y_{r_i}^2 - z_{r_i}^2  \, \, \quad i=1,\ldots, N,  \\
b_0 = \sqrt{x_p^2+y_p^2+z_p^2}
\end{cases}
\end{equation}
where it is assumed $b_i \geq 0$, $i=1,\ldots,N$.
All the relationships described in \eqref{model_equation_2} can be grouped in a more compact matrix form as

\begin{equation}\label{model_equation_3}
\begin{cases}
\bH {\bp} - \bg = {\bf 0}\\
\bp^T \bp = {b^2_0}\\
\end{cases}
\end{equation}

where 
\begin{align*}
&\bH^T = \left[\textit{h}_1,\textit
{h}_2,\ldots,\textit{h}_N\right] \in \mathbb{R}^{3,N}, \mbox{with} \\
&\textit{h}_i = \left[-2 x_{r_i}, -2 y_{r_i},-2 z_{r_i} \right]^T \in \mathbb{R}^3\,,\,\,i=1,\ldots, N\\
&\bg = \left[\begin{matrix} g_1, \ldots, g_N \end{matrix}\right]^T \in \mathbb{R}^N, \mbox{with} \\
&\textit{g}_i = b_i^2-b_0^2 -  x_{r_i}^2 - y_{r_i}^2 - z_{r_i}^2  \,,\,\,i=1,\ldots, N.    
\end{align*}

\subsection{Monostatic Acquisition System and Target Position Constraints}\label{MB_constraint}
To perform the measurement process, the active radar employs an antenna characterized by a specific transmit/receive beampattern with a given main-lobe width and pointing direction, without loss of generality, coincident with the $x$-axis of the reference system. In this subsection, some constraints able to capitalize such a-priori information are formalized with the goal of improving target positioning reliability. To this end, let us denote by:
\begin{itemize}
\item $\bar{\theta}$ and $\bar{\phi}$ the  (half-side) antenna beamwidths {in the $x-y$ and $x-z$ plane, respectively}, as shown in Fig. \ref{fig:beam};
\item $\theta_p = \text{atan2}(y_p,x_p)$ and $\phi_p = \text{atan2}(z_p,x_p)$ the {azimuth (in the $x-y$ plane) and elevation (in the $x-z$ plane)} target angular coordinates, respectively. 
\end{itemize}
\begin{figure}[!h]
    \centering
    \includegraphics[width=0.7\columnwidth]{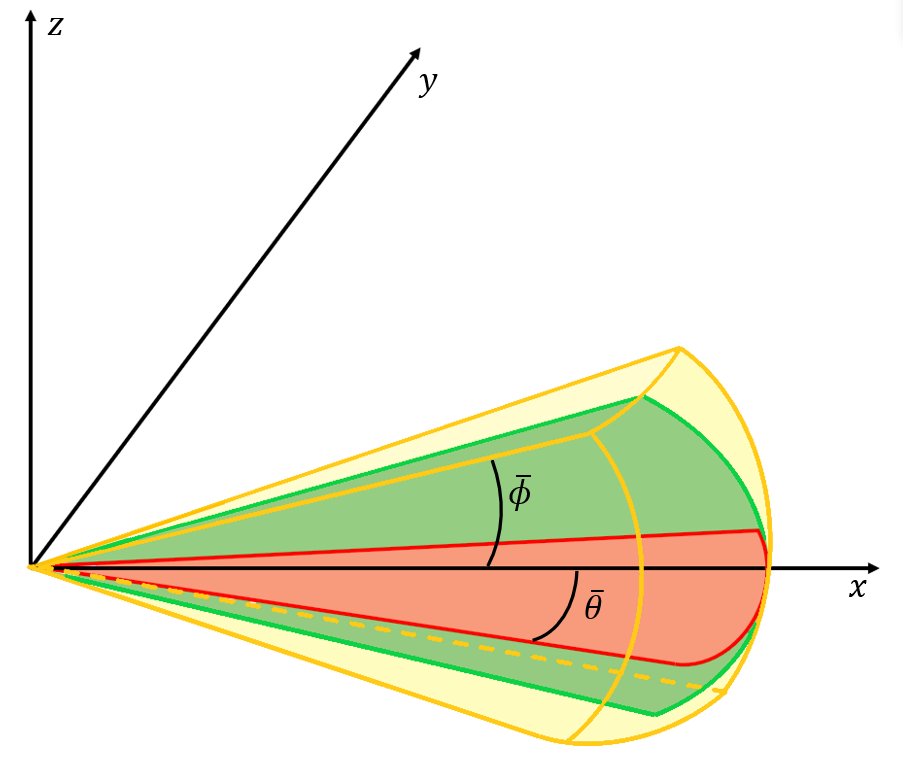}
        \caption{Representation of the antenna beamwidth.}
    \label{fig:beam}
\end{figure}
Hence, let us observe that the limited main-lobe extension of the active radar demands the angular location of any illuminated target to comply with
\begin{align}\label{eq_ran1}
\begin{cases}
-\bar{\theta} \leq \theta_p \leq \bar{\theta}\\
-\bar{\phi} \leq \phi_p \leq \bar{\phi}
\end{cases}.
\end{align}
The relationship \eqref{eq_ran1} coupled with the assumptions $0 \leq \bar{\theta} < \frac{\pi}{2}$ and 0 $\leq \bar{\phi} < \frac{\pi}{2}$, can be equivalently rewritten as

\begin{align}\label{eq_ran2}
\begin{cases}
-\tan \bar{\theta} \leq \tan \left( \theta_p \right) \leq \tan \bar{\theta}\\
-\tan \bar{\phi} \leq \tan \left( \phi_p \right) \leq \tan \bar{\phi}
\end{cases},
\end{align}
which boils down to
\begin{equation}\label{eq:constraintangle2}
\begin{cases}
- x_{p} \gamma_a \leq y_{p} \leq x_{p} \gamma_a\\
- x_{p} \gamma_a \leq z_{p} \leq x_{p} \gamma_e\\
x_{p} \geq 0\\
\end{cases},
\end{equation}
where $\gamma_a = \tan \bar{\theta}$ and $\gamma_e = \tan \bar{\phi}$.
It is worth pointing out that, with reference to a 2D problem, angular constraints have been also used in \cite{acls} for target localization with a passive radar  and in \cite{aacls,AACLS_journal} to realize a passive radar positioning aided by some a-priori information provided by an active radar sensing system.
\raggedbottom

\section{Problem Formulation and 3D Localization Algorithm}\label{Problem Formulation}
This section deals with the formalization of the localization problem and the development of the resulting estimation technique. To this end, both model equations \eqref{model_equation_3} and constraints \eqref{eq:constraintangle2}, induced by the monostatic acquisition system, are exploited. In this respect, it is important to highlight that the matrix $\bH$, the vector $\bg$, and the target range $b_0$ are corrupted by noise, and, consequently, the relationships in \eqref{model_equation_3} are not exactly satisfied. This issue is handled resorting to the constrained LS framework, forcing the target range to be the projection of the range measurement within the detected range-bin, and looking for the
best fitting of the model to observations according to a squared norm cost function. Specifically, indicating the mentioned projection by $\bar{b}_0=\max(\min(b_0,r_U), r_L)$ with $r_L$ and $r_U$ the extremes of the detected range-bin, the target positioning process can be formalized as the following non-convex optimization problem

\begin{subequations}\label{eq_prob_comp2}
\begin{align}[left = {\cal P} \empheqlbrace]
{\displaystyle \min_{{\bp}}} & {\left\| \bH {\bp} - \bg \right\|}^2\\
\mbox{s.t.} & \|\bp\|^2= \bar{b}_0^2 \label{eq:constraint1} \\ 
            & - {x}_p \gamma_a \leq {y}_p \leq {x}_p \gamma_a \label{eq:constraint2}\\ 
            & - {x}_p \gamma_e \leq {z}_p \leq {x}_p \gamma_e \label{eq:constraint3}\\ 
            & {x}_p \geq 0 \label{eq:constraint4}
\end{align}
\end{subequations}
 Although {$\cal P$} is difficult to solve, through the use of optimization techniques based on KKT optimality \cite{Bertsekas16} conditions, a quasi-closed-form (i.e., whose computation involves only elementary functions and roots of polynomial equations) optimal solution can be derived.
Indeed, the following proposition holds true.

\begin{proposition}\label{Prop1}
An optimal solution to ${\cal P}$ belongs to the following finite set of feasible points (whose cardinality is at
most twenty-six):
\begin{enumerate}

\item ${\bx}^\star(\bar\lambda_h) = \left({\bH}^T{\bH} + \bar\lambda_h \bI\right)^{-1}{\bH}^T\bg, ~h \in I_1 \subseteq \{1,\ldots,6\}$, with $\bar\lambda_h$ the real-valued solutions to the sixth-order equation

\begin{equation}\label{6_order_equation}
{\bx}^\star(\bar\lambda)^T {\bx}^\star(\bar\lambda) =\bar{b}_0^2
\end{equation}

such that

\begin{equation} \label{constraint_beck1}
\begin{cases}
- \gamma_a {x}_p^\star(\bar\lambda_h) < {y}_p^\star(\bar\lambda_h) < \gamma_a  {x}_p^\star(\bar\lambda_h)\\
- \gamma_e {x}_p^\star(\bar\lambda_h) < {z}_p^\star(\bar\lambda_h) <\gamma_e  {x}_p^\star(\bar\lambda_h)\\
{x}_p^\star(\bar\lambda_h ) > 0
\end{cases}.
\end{equation}
\item $\bx^*(\beta^i_h) =[q_1^*(\beta^i_h), (-1)^{i+1} \gamma_a q_1^*(\beta^i_h), q_2^*(\beta^i_h)]^T, i =1,2$ with

\begin{equation}
    \bq^*(\beta^i_h) = \left( {\bH_i^a}^T{\bH_i^a} + \beta^i_h \bB^a\right)^{-1}{\bH_i^a}^T\bg
\end{equation}

where 
\begin{equation}
  {\bH_i^a} = \bH \left[\begin{matrix} 1 & 0 \\
(-1)^{i+1}\gamma_a & 0 \\
0 & 1
\end{matrix}\right], i=1,2,   
\end{equation}
\begin{equation}
 \bB^a  = \left[\begin{matrix} 1 +\gamma_a^2 & 0 \\0 & 1 \end{matrix}\right],
\end{equation}

and $\beta^i_h, ~h\in I^i_2 \subseteq \{1,\ldots,4\}  $ the real-valued solutions to the fourth-order equation 

\begin{equation}\label{4_order_equation_1}
  {\bq^*}^T(\beta^i) \bB^a \bq^*(\beta^i) = \bar{b}_0^2  
\end{equation}

such that 

\begin{equation} \label{constraint_beck2}
\begin{cases}
- \gamma_e {q}_1^\star(\beta^i_h) < {q}_2^\star(\beta^i_h) <\gamma_e  {q}_1^\star(\beta^i_h)\\
{q}_1^\star(\beta^i_h) > 0
\end{cases}.
\end{equation}

\item $\bx^*(\eta^i_h) =[p_1^*(\eta^i_h), p_2^*(\eta^i_h), (-1)^{i+1}\gamma_e p_1^*(\eta^i_h)]^T, i =1,2$ with

\begin{equation}
    \bp^*(\eta_h^i) = \left( {\bH_i^e}^T\bH^e_i + \eta^i_h \bB^e\right)^{-1}{\bH_i^e}^T \bg
\end{equation}

where  
\begin{align}
\bH^e_i = \bH \left[\begin{matrix} 1 & 0 \\
 0 & 1 \\
(-1)^{i+1}\gamma_e & 0
\end{matrix}\right], i= 1,2,
\end{align}
\begin{equation}
   \bB^e = \left[\begin{matrix} 1 +\gamma_e^2 & 0 \\0 & 1 \end{matrix}\right],
\end{equation}

and $\eta^i_h, ~h\in I^i_3 \subseteq \{1,\ldots,4\}$ the real-valued solutions to the fourth-order equation 

\begin{equation}\label{4_order_equation_2}
  {\bp^*}^T(\eta^i) \bB^e \bp^*(\eta^i) = \bar{b}_0^2
\end{equation}

such that

\begin{equation} \label{constraint_beck3}
\begin{cases}
- \gamma_a {p}_1^\star(\eta^i_h) < {p}_2^\star(\eta^i_h) <\gamma_a  {p}_1^\star(\eta^i_h)\\
{p}_1^\star(\eta^i_h) > 0
\end{cases}.
\end{equation}

\item $
  \bx^*_{4_{i,j}} =\alpha \left[ 1, (-1)^{1+i} \gamma_a ,  (-1)^{1+j} \gamma_e  \right]^T, (i, j) \in \{1, 2\}^2,$ with $\alpha = \frac{\bar{b}_0}{\sqrt{1+ \gamma_a^2 + \gamma_e^2}}.$

\end{enumerate}
 \end{proposition}

\begin{proof}
See Appendix \ref{AppendixProp1}
\end{proof}

In a nutshell, proposition \ref{Prop1} defines the optimal candidate solutions to
Problem ${\cal P}$. Precisely, each subset of solutions refers
to a specific portion of the feasible target locations. A complete description of the global optimum search procedure is reported in Algorithm \ref{final_algorithm}. 
It is worth observing that the determination of each subset of candidates requires the evaluation of the roots of a specific polynomial equation, whose efficient computation is discussed in the following subsection.

\begin{algorithm}[!h]
\caption{Localization Algorithm}
\label{final_algorithm}
\textbf{Input:} $\bar\theta, \bar\phi, \tau_i,  {\bp_r}_i,~i=0,\ldots,N, \epsilon$. \\
\textbf{Output:} Target position estimate $[\hat x_p, \hat y_p, \hat z_p]^T$;
\begin{enumerate}
    \item {{\bfseries{Parameter setup:}}}
   {\small{ \begin{itemize}
        \item Compute ~$\bH, \bB^{a}, \bB^{e}, \bH^{a}_i, \bH^{e}_i, ~i=1,2$, $\bar\bb_0$,  $\gamma_a,\gamma_e,\alpha$.
    \end{itemize}}}
          \item {{\bfseries{Candidate points evaluation}}} 
        \small  \begin{itemize}
             { 
              
              \item Find the roots (with the accuracy $\epsilon$) $\bar\lambda_h, h\in{I_1\subseteq\{1,2,\ldots,6\}}$ of equation \eqref{6_order_equation} satisfying \eqref{constraint_beck1}; then, compute $\bx^{*}(\bar{\lambda}_h)$;
                          
              \item Find the roots (with the accuracy $\epsilon$) $\beta^{i}_h, h\in{I_2^i\subseteq\{1,2,3,4\},~i=1,2}$ of equation \eqref{4_order_equation_1} satisfying \eqref{constraint_beck2}; then, compute $\bx^{*}(\beta^{i}_h)$;
              
              \item Find the roots (with the accuracy $\epsilon$) $\eta^{i}_h, h\in{I_3^i\subseteq\{1,2,3,4\},~i=1,2}$ of equation \eqref{4_order_equation_2} satisfying \eqref{constraint_beck3}; then, compute $\bx^{*}(\eta^{i}_h)$;
              \item Compute ${\bx^*_{4_{i,j}}}, (i,j)\in \{1,2\}^2$.}
          \end{itemize}
          
    \item {{\bfseries{Optimal Solution Selection}}}
         \begin{itemize}
          
             \item {\small Compute:}
             {\footnotesize{\begin{itemize}
                 \item $v_j=\|\bH\bx^{*}(\bar\lambda_{j})-\bc\|, ~j=1,\ldots,|I_1|$
                  
                  \item $v_{|I_1|+j}=\|\bH\bx^{*}({\beta^{1}_{j}})-\bc\|, ~j=1,\ldots,|{I^{1}_2}|$
 
                   \item $v_{|I_1|+|{I^{1}_2}|+j}=\|\bH\bx^{*}({\beta^{2}_{j}})-\bc\|, ~j=1,\ldots,|{I^{2}_2}|$

                  \item $v_{|I_1|+|{I^{1}_2}|+|{I^{2}_2}|+j}=\|\bH\bx^{*}({\eta^{1}_{j}})-\bc\|, ~j=1,\ldots,|{I^{1}_3}|$
 
                   \item $v_{|I_1|+|{I^{1}_2}|+|{I^{2}_2}|+|{I^{1}_3}|+j}=\|\bH\bx^{*}({\eta^{2}_{j}})-\bc\|, ~j=1,\ldots,|{I^{2}_3}|$
                  
                 \item $v_{|I_1|+|{I^{1}_2}|+|{I^{2}_2}|+|{I^{1}_3}|+|{I^{2}_3}|+j}=\|\bH\bx^*_{4_{j,1}}-\bc\|, ~j=1,2$
                
                 \item $v_{|I_1|+|{I^{1}_2}|+|{I^{2}_2}|+|{I^{1}_3}|+|{I^{2}_3}|+2+j}=\|\bH\bx^*_{4_{j,2}}-\bc\|, ~j=1,2$.              
   \end{itemize}}}

   \item {\small Determine} $j^*=\mbox{arg}~\displaystyle{\min_{j}}~v_j$ and pick up the corresponding solution, i.e.,
     \end{itemize} 
 {\footnotesize{   \begin{align*}
             \bar{\bx}^{*}=\begin{cases}
       \bx^{*}(\bar{\lambda}_{j^*})\mbox{ if}\mbox{ }1\leq j^{*} \leq |I_1|\\
       {\bx^{*}}({\beta^{1}_{j^*}}) ~  \mbox{if} ~ j^\star \geq |I_1|+1 ~\&\\
       ~~~~~~~~~~~~~~ j^\star \leq |I_1|+ |I_2^1|\\
      {\bx^{*}}({\beta^{2}_{j^*}}) ~  \mbox{if} ~ j^\star \geq |I_1|+ |I_2^1|+1 ~\&\\
      ~~~~~~~~~~~~~~ j^\star \leq |I_1|+ |I_2^1|+ |I_2^2|\\
       {\bx^{*}}({\eta^{1}_{j^*}}) ~ \mbox{if}  ~ j^\star \geq |I_1|+ |I_2^1|+ |I_2^2|
       +1~\&\\
       ~~~~~~~~~~~~~~ j^\star \leq |I_1|+ |I_2^1|+ |I_2^2|+ |I_3^1|\\
      {\bx^{*}}({\eta^{2}_{j^*}}) ~ \mbox{if} ~ j^\star \geq |I_1|+ |I_2^1|+ |I_2^2|+ |I_3^1|+1~\&\\
      ~~~~~~~~~~~~~~ j^\star \leq |I_1|+ |I_2^1|+ |I_2^2|+ |I_3^1|+ |I_3^2|\\
    {\bx^*_{4_{j^*,1}}} ~~~ \mbox{if} ~  j^\star \geq |I_1|+ |I_2^1|+ |I_2^2|+ |I_3^1|+ |I_3^2|+1~\&\\ 
     ~~~~~~~~~~~~~~ j^\star\leq |I_1|+ |I_2^1|+ |I_2^2|+ |I_3^1|+ |I_3^2|+2\\
      {\bx^*_{4_{j^*,2}}}~~~~~~~~~~~~~ \mbox{otherwise}
       \end{cases}
   \end{align*}}}

    {{\bfseries{Output:}}} $[\hat{x}_p, \hat{y}_p, \hat{z}_p]^{T}=[{\bar{x}^{*}}_1, {\bar{x}^{*}}_2,{\bar{x}^{*}}_3]^{T}$.
          
\end{enumerate} 
\end{algorithm}

\subsection{{Evaluation of the} Roots and Algorithm Computational Complexity}
Algorithm \ref{final_algorithm} involves the solution of  equations \eqref{6_order_equation}, \eqref{4_order_equation_1}, and \eqref{4_order_equation_2}. Guidelines and insights to the rooting process are now provided with reference to equation\footnote{Analogous considerations hold true for equations \eqref{4_order_equation_1} and \eqref{4_order_equation_2}.} \eqref{6_order_equation}. As shown in Appendix \ref{appendix:sol_det} (the interested reader may refer to it for technical details and parameters definitions) solving \eqref{6_order_equation} is tantamount to finding the real-valued roots\footnote{In Appendix \ref{appendix:sol_det}, a  normalized version of \eqref{6_order_equation_recast_1} is analyzed.} of
\begin{equation}\label{6_order_equation_recast_1}
\sum_{j=1}^3 \frac{|z_j|^2}{(\bar\lambda + \lambda_j)^2}-\bar b_0^2.
\end{equation}
Evidently, each root of \eqref{6_order_equation_recast_1} must belong to one of the four subsets
$\mathcal{J}_1 = (-\infty,-\lambda_3),$
$\mathcal{J}_2 = (-\lambda_3, -\lambda_2),$
$\mathcal{J}_3 = (-\lambda_2, -\lambda_1),$ and
$\mathcal{J}_4 = (-\lambda_1,+\infty)$.

Now, being \eqref{6_order_equation_recast_1} strictly increasing (decreasing) over $\mathcal{J}_1 $ ($\mathcal{J}_4$) with a range $(-\bar b_0^2,+\infty)$, a unique root exists within $\mathcal{J}_1$ ($\mathcal{J}_4$) and it can be found through the standard bisection method.

In $\mathcal{J}_2$ ($\mathcal{J}_3$), instead, zero, one, or even two roots can exist, depending on the range of \eqref{6_order_equation_recast_1} over $\mathcal{J}_2$ ($\mathcal{J}_3$). Leveraging the strict convexity of \eqref{6_order_equation_recast_1}, these points can be determined according to a two-stage process involving at most three bisection loops each of them applied to either \eqref{6_order_equation_recast_1} or its derivative.

As discussed in Appendix \ref{appendix:sol_det}, the parameters of \eqref{6_order_equation_recast_1} can be computed via elementary functions applied to the entries of $(\bH^T\bH)$, whose evaluation involves $\mathcal{O}(N^2)$ operations. Now, denoting by ${\epsilon}_0$ the maximum size (among the different bisections) of the initial search interval (see Appendix \ref{equation_resolution} for their determination) and by $\epsilon$ the desired accuracy level of any root, the number $n$ of bisection iterations, in each search process, is upper bounded by
\begin{equation}
n = \lceil\log_2\left(\frac{\epsilon_0}{\epsilon}\right) \rceil.
\end{equation}
Finally, each bisection cycle is performed with a computational complexity of $\mathcal{O}(1)$, being involved just elementary functions and comparisons. Hence, for a given accuracy $\epsilon$, the roots search process entails $\mathcal{O}(1)$ operations, given $(\bH^T \bH)$. It is worth pointing out that similar conclusions apply to
the solution of \eqref{4_order_equation_1} and \eqref{4_order_equation_2}.

Let us now deal with the computational complexity of {\bf Algorithm 1}. Given the solutions to \eqref{6_order_equation}, \eqref{4_order_equation_1}, and \eqref{4_order_equation_2}, it mainly entails a) the evaluation of the resulting candidate optimal solutions and b) the computation of the corresponding objective values. The former can be accomplished with a computational burden of $\mathcal{O}(1)$, being embroiled elementary functions comparisons and the matrices $(\bH^T\bH)$, $({\bH_i^a}^T{\bH_i^a})$, and $({\bH_i^e}^T{\bH_i^e})$ are already computed in the bisection processes. The latter requires $\mathcal{O}(N)$ operations to evaluate the squared norms. As a result, the overall computational complexity of Algorithm \ref{final_algorithm} is $\mathcal{O}(N^2)$.

\section{Performance Analysis}\label{analysis}
\begin{figure*}[!htb]
    \centering
    \includegraphics[width=1\textwidth]{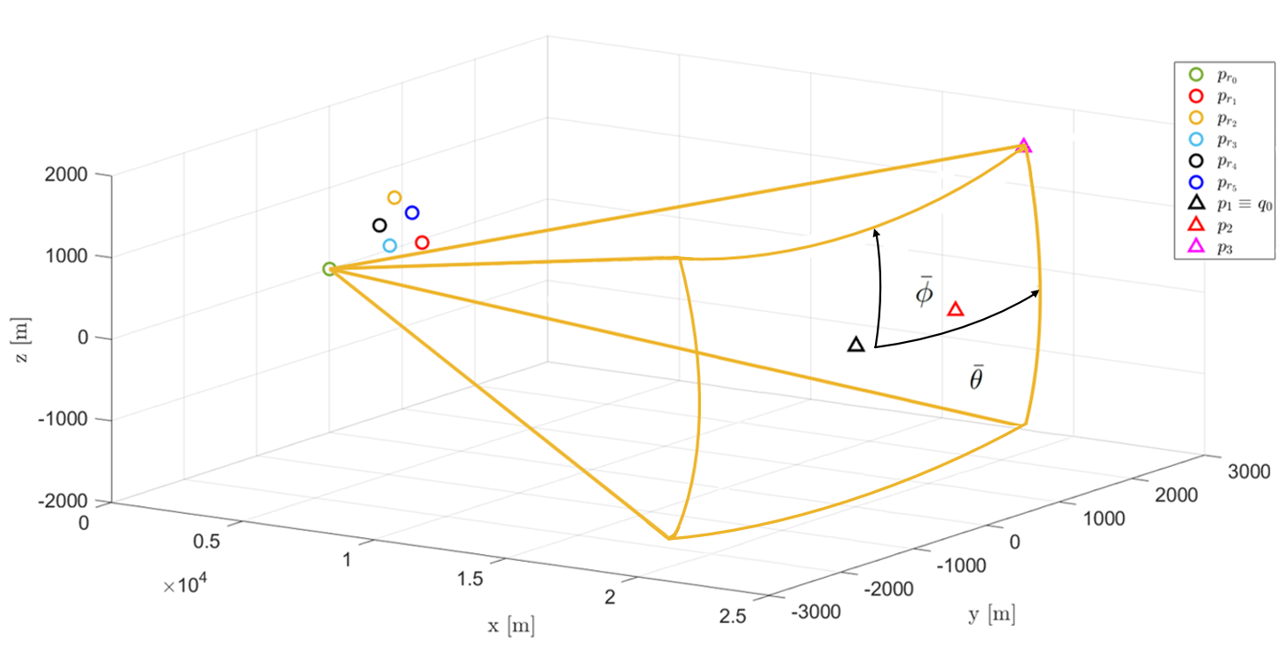}
        \caption{Geometric configuration of the radiolocation system and target location scenarios. In the figure legend, $\bp_1$, $\bp_2$, and $\bp_3$ are the considered target positions, located at range $r = 20$ km and with azimuth-elevation $(0^\circ, 0^\circ), (4^\circ, 0^\circ)$, and $(6.9^\circ, 4.9^\circ)$, respectively.}
    \label{fig:setup_sensors}
\end{figure*}
Radiolocation systems composed of
$N = 4$ and $N=5$ receive-only sensors and one active radar are considered. Focusing first on the $N=4$ case, the  receiving nodes are located at ${\bp}_{r_1} = [916,  941,  95]^T$ km, ${\bp}_{r_2} = [973,541,764]^T$ km, ${\bp}_{r_3} = [955,483, 191]^T$ km, and ${\bp}_{r_4} = [936,350,477]^T$ km.
\begin{figure}[!htb] \centering
\subfigure[$\theta_p = 0^\circ, \phi_p = 0^\circ $.]
{\includegraphics[width=0.9\columnwidth]{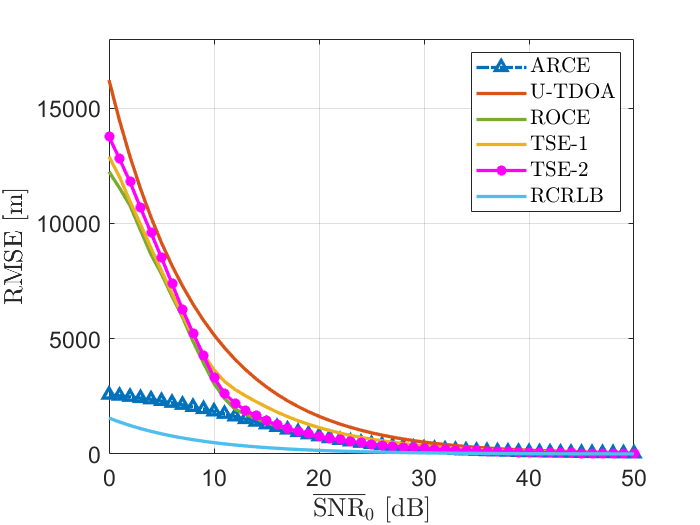}}\par\medskip
\subfigure[$\theta_p = 4^\circ, \phi_p = 0^\circ$.]
{\includegraphics[width=0.9\columnwidth]{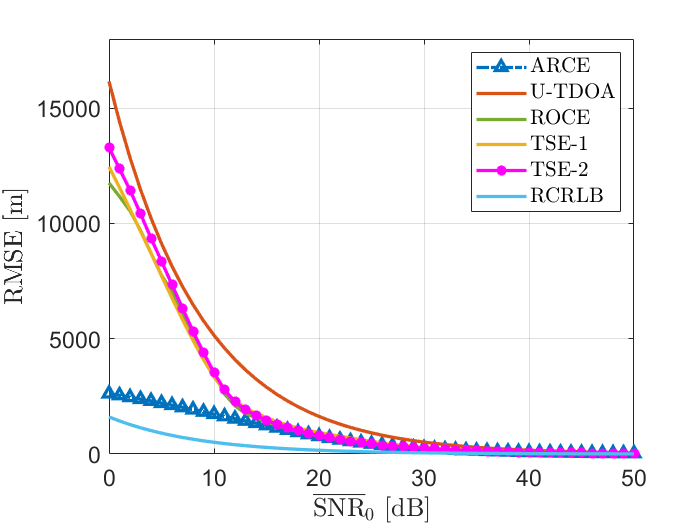}}\par\medskip
\subfigure[$\theta_p = 6.9^\circ, \phi_p = 4.9^\circ$]
{\includegraphics[width=0.9\columnwidth]{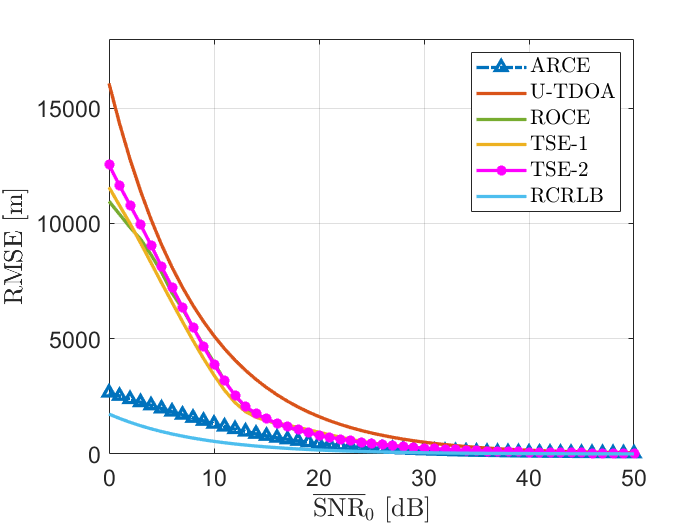}}
\caption{RMSE versus $\overline{\mbox{SNR}}_0$, when $\bar\theta=7^\circ$, $\bar\phi=5^\circ$, and the radiolocation system comprises $N = 4$ receive-only  sensors.}
\label{fig:RMSE_vs_SNR_4_sensors}
\end{figure}
Moreover, the measurement errors are modeled as zero-mean independent Gaussian random variables with standard deviations given in \eqref{variance_bistatic_measurements}. Therein, the SNR of the $N = 4$ bistatic pairs (transmitter-receiver) and
of the active radar are calculated as 

\begin{equation}
\mbox{SNR}_i=\frac{{\overline{\mbox{SNR}}}_0}{L_i}\frac{\|\bq_0\|^2}{\|\bp\|^2}\frac{\|\bq_0\|^2 }{\|\bp-{\bp_r}_i\|^2}, i=0,1,\ldots,4       
\end{equation}
where $\overline{\mbox{SNR}}_0$ is a reference SNR computed via the monostatic radar range equation \cite{Richards} at the nominal point $\bq_0 = [20, 0, 0]^T$ km and $L_i, i=0,\ldots,4$, accounts for a loss factor due to different receive gains of the active sensor and the receive-only units. In particular, $L_0 = 0$ dB, while $L_i=6$ dB, $ i=1,\ldots,4$.

\begin{figure}[!htb] \centering
\subfigure[$\theta_p = 0^\circ, \phi_p = 0^\circ $.]
{\includegraphics[width=0.9\columnwidth]{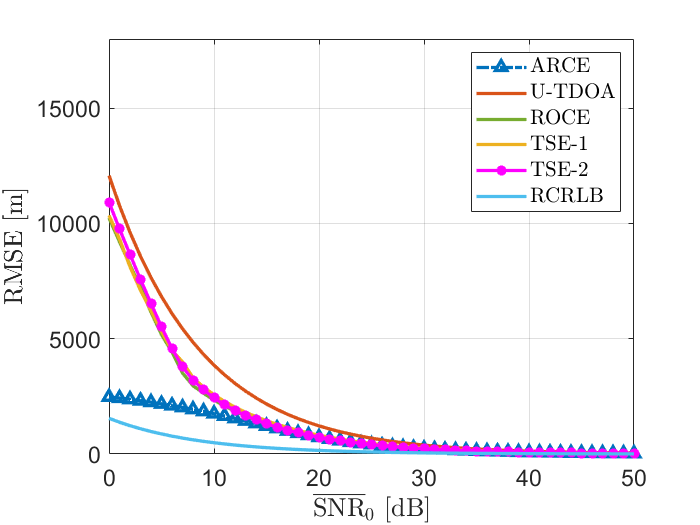}}\par\medskip
\subfigure[$\theta_p = 4^\circ, \phi_p = 0^\circ$.]
{\includegraphics[width=0.9\columnwidth]{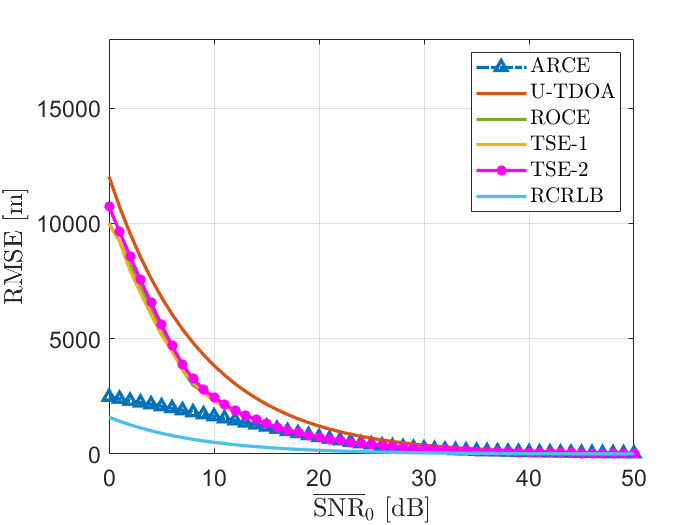}}\par\medskip
\subfigure[$\theta_p = 6.9^\circ, \phi_p = 4.9^\circ$]
{\includegraphics[width=0.9\columnwidth]{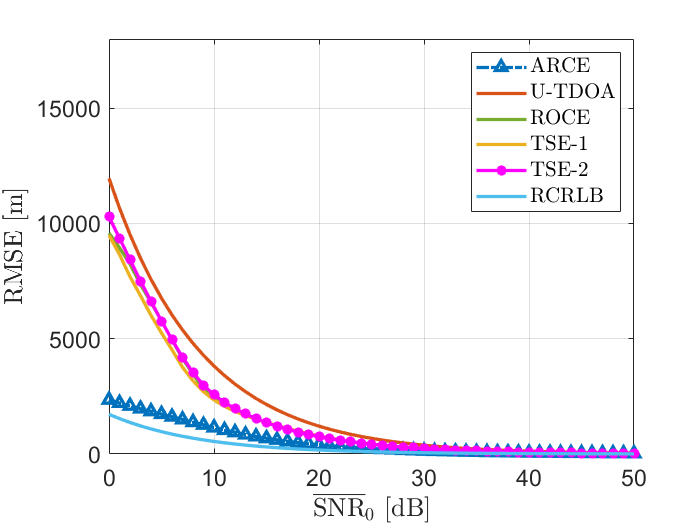}}
\caption{RMSE versus $\overline{\mbox{SNR}}_0$, when $\bar\theta=7^\circ$, $\bar\phi=5^\circ$, and the  radiolocation system comprises $N = 5$ receive-only  sensors.}
\label{fig:RMSE_vs_SNR_5_sensors}
\end{figure}

\begin{figure}[!htb] \centering
{\includegraphics[width=0.9\columnwidth]{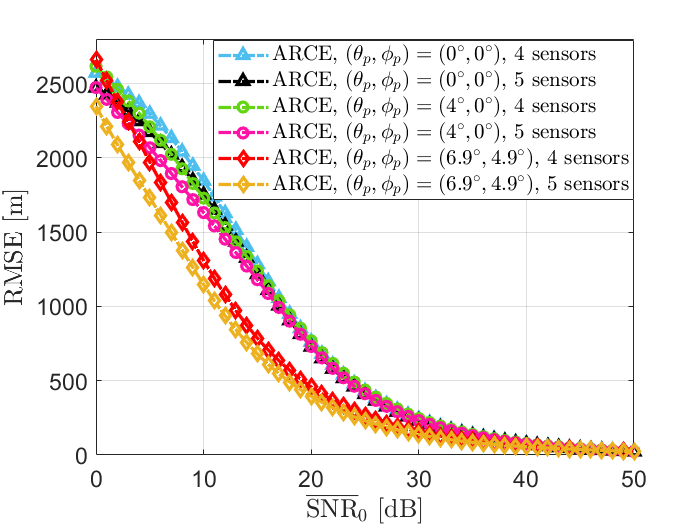}}
\caption{RMSE versus ${\overline{\mbox{SNR}}}_0$, when $\bar\theta=7^\circ$, $\bar\phi=5^\circ$, and radiolocation system comprises $N = 4$ and $N = 5$ receive-only sensors.}
\label{fig:comparison}
\end{figure}

\begin{figure}[!htb] \centering
{\includegraphics[width=0.9\columnwidth]{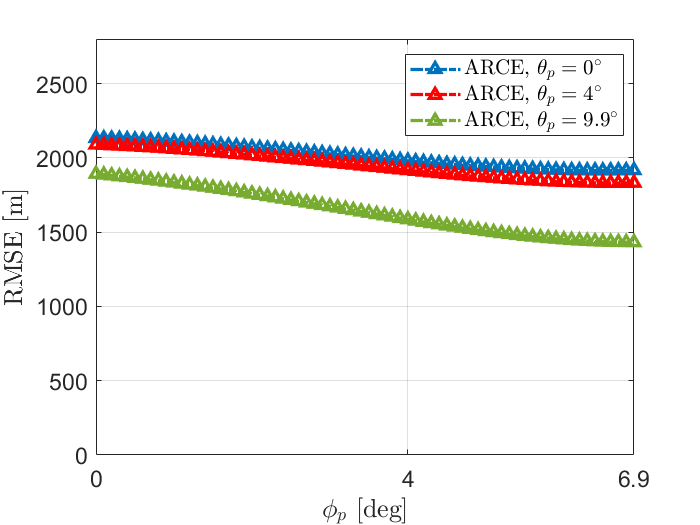}}
\caption{RMSE versus $\phi_p$, for $\overline{\mbox{SNR}}_0 = 10$ dB, when  $\bar\theta=10^\circ$, $\bar\phi=7^\circ$, and radiolocation system comprises $N = 4$ receive-only sensors.}
\label{fig:RMSE_vs_angle}
\end{figure}

\begin{figure}[!htb] \centering
{\includegraphics[width=0.9\columnwidth]{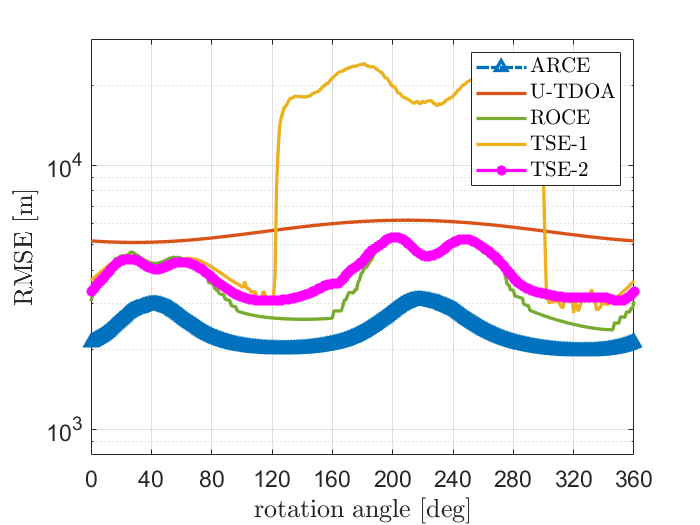}}
\caption{RMSE versus radar pointing direction, for $\overline{\mbox{SNR}}_0 = 10$ dB, when $\bar\theta=10^\circ$, $\bar\phi=7^\circ$, and radiolocation system comprises $N = 4$ receive-only sensors.}
\label{fig:rotating_radar}
\end{figure}

The performance of the developed localization algorithm is evaluated considering as a figure of merit the RMSE of the target position estimate,
formally defined as
$\sqrt{E[\|\hat\bp -  \bp\|^2]}$, where $\hat\bp$ is the estimated position. Since the RMSE does not present a closed-form expression, Monte Carlo simulation with $1000$ independent trials is exploited. 
Additionally,  the Root CRLB (RCRLB), defined as $\sqrt{\mbox{tr} (\mbox{FIM}^{-1})}$,
where FIM denotes the Fisher Information Matrix associated with the unknown parameters \cite{vantrees}, is provided as performance benchmark. For comparison purposes, also the performance of some counterparts are illustrated. Specifically, the performance of the procedures developed\footnote{Their implementation for the 3D case with a transmitter co-located with one of the receivers at a known location is considered.}  in \cite{TSE} and \cite{Multistatic_joint_localization}, denoted hereafter as TSE-$1$ and TSE-$2$, respectively, is reported too, along with two alternative methods for target positioning based on the measurement model in \eqref{bistatic_measurement}, discussed in the following.
\raggedbottom
\begin{itemize}
\item {U-TDOA-Like Estimator}\label{totally_unconstrained_estimator}
Unconstrained U-TDOA resorts to a standard  LS framework and, assuming $\bH$ full-rank, gives the following position estimate
\begin{equation}
    \hat{\bp}_{TDOA} = \arg \min_{\bp\in\mathbb{R}^3}{\|\bH\bp-c\|^2}=\left(\bH^T\bH\right)^{-1}\bH^T\bg.
\end{equation}

\item
{Range-Only Constrained Estimator (ROCE)}\label{partially_constrained_estimator}
Giving up the main-lobe constraint, the estimation problem can be framed as
\begin{subequations}
\begin{align}[left = {\cal P} \empheqlbrace]
{\displaystyle \min_{{\bp}}} & {\left\| \bH {\bp} - \bg \right\|}^2\\
\mbox{s.t.}  &\|\bp\|^2= \bar b_0^2 
\end{align}
\end{subequations}
and the resulting position can be retrieved as
\begin{equation}
 \hat{\bx}=\arg\min_{i\in I_1} {\left\| \bH {\bx_i} - \bg \right\|}^2
\end{equation}
where 
\begin{equation}
\bx_i = \left({\bH}^T{\bH} + \bar\zeta_i \bI\right)^{-1}{\bH}^T\bg, ~i \in I_1 \subseteq \{1,\ldots,6\},  
\end{equation}  with $\bar\zeta_i$ the real-valued solutions to the sixth-order equation 
\begin{equation}
{\bx}^\star(\bar\zeta)^T {\bx}^\star(\bar\zeta) = \bar b_0^2.
\end{equation}

\end{itemize}

 In the considered numerical analysis, the target is positioned
at $[r \cos{\theta_p} \cos{\phi_p}, r \sin{\theta_p} \cos{\theta_p}, r\sin{\phi_p}]^T$, with
$r = 20$ km and different values of $\theta_p$ and $\phi_p$ are considered, i.e., $(\theta_p, \phi_p) \in\{(0^\circ, 0^\circ), (4^\circ, 0^\circ), (6.9^\circ, 4.9^\circ)\}$. Furthermore, the main-beam width in azimuth
 and elevation for the monostatic radar are $\bar\theta = 7^\circ$ and $\bar\phi = 5^\circ$, respectively.
 {Finally, the transmit signal bandwidth equal to $B = 2$ MHz.} The considered target positions along with the radiolocation system and the radar main-lobe radiation pattern are displayed in  Fig. \ref{fig:setup_sensors}.

In Figs. \ref{fig:RMSE_vs_SNR_4_sensors}, the RMSE versus ${\overline{\mbox{SNR}}}_0$ is illustrated, where each subfigure refers to a specific scenario for the target position. The results presented in Fig.\ref{fig:RMSE_vs_SNR_4_sensors} show that the
designed estimator achieves a satisfactory performance being the corresponding RMSE curves closer and closer to the CRLB, as the SNR increases. 
Moreover, the devised algorithm outperforms the counterparts, e.g., U-TDOA estimator, ROCE estimator, TSE$-1$ and TSE$-2$, with interesting gains in the range $[0,20]$ dB for $\overline{\mbox{SNR}}_0$. This behaviour corroborates the advantage of using a-priori information related to the characteristics of the monostatic antenna beamwidth to improve the localization performance.
Specifically, comparing Figs. \ref{fig:RMSE_vs_SNR_4_sensors}(a)-(b)-(c) it is evident that the U-TDOA estimator and TSE$-2$ achieve a performance comparable with the ARCE for ever greater  values of $\overline{\mbox{SNR}}_0$ as the  target distance from  the antenna’s pointing direction increases. This behaviour is not surprising since the beampattern extent constraint is more valuable when the target is closer and closer to the main-lobe boundary. 
These insights are confirmed by the plots  in Fig. \ref{fig:RMSE_vs_SNR_4_sensors}(c) pinpointing that the proposed estimator achieves its best performance when $(\theta_p, \phi_p)=(6.9^\circ, 4.9^\circ)$. 
\raggedbottom

In Fig. \ref{fig:RMSE_vs_SNR_5_sensors} the RMSE is also analyzed assuming the scenario of Fig. \ref{fig:RMSE_vs_SNR_4_sensors} and an additional receive unit located at ${\bp_r}_5 = [760,860,477]^T$ km. Inspection of the figures corroborates the merits of the proposed algorithm with respect to the counterparts. As expected, all the considered procedures provide better estimates than the case of Fig. \ref{fig:RMSE_vs_SNR_4_sensors} with $N = 4$. To gather further insights about the impact of the number of receivers on the proposed technique, in Fig. \ref{fig:comparison} the RMSE of the devised localization method is displayed versus the $\overline{\mbox{SNR}}_0$, for $N = 4$ and $N = 5$. 
As expected, the results in Fig. \ref{fig:comparison} highlight that the presence of the additional receiver can grant a performance improvement ranging between 100 and 400 meters, for ${\overline{\mbox{SNR}}}_0$ values from $0$ to $20$ dB.

{ The case of a wider  main-lobe width is considered in Fig. \ref{fig:RMSE_vs_angle}, i.e., $\bar\theta = 10^\circ$ and $\bar\phi = 7^\circ$. 
The performance in terms of RMSE for the ARCE estimator is plotted versus the elevation angle $\phi_p$, for three values of the target azimuth.
The curves show that the proposed algorithm provides more accurate estimates 
 when $\theta_p=9.9^\circ$, as compared with $\theta_p=0^\circ$ and $\theta_p=4^\circ$. It is worth mentioning
that the performance improves as $\phi_p$ increases regardless of $\theta_p$. This behavior confirms that the estimation error decreases when the target is closer to the boundary of the angular region.

In Fig. \ref{fig:rotating_radar} the localization capabilities are analyzed when the radar antenna pointing direction rotates in the $x-y$ plane, assuming $\overline{\mbox{SNR}}_0 = 10$ dB. 
The target is supposed  fixed and centered in the radar beam, regardless of the pointing direction. This case study allows to evaluate one of the main skills of interest that such a multistatic radar network presents, i.e., the geometric diversity provided by the system, which depends on the spatial configuration of the receiver nodes with respect to the active radar main beam direction.  The curves in Fig. \ref{fig:rotating_radar} pinpoint that ARCE strategy is sensibly more accurate and robust than the alternative methods, especially in comparison with TSE-$1$.} Indeed, the improvement is observed for all the rotation angles.

\section{Conclusions}\label{conclusions}
A novel strategy for 3D target positioning has been developed for C-RFMs composed of a master transmit-receive node and multiple receive sensors. The monostatic radiation pattern features have been wisely exploited in the proposed positioning process restricting the angular location of any illuminated target. Hence, leveraging monostatic and bistatic range measurements, the Cartesian coordinates of the target have been estimated as the global optimal solution to a constrained non-convex LS problem. Resorting to the KKT optimality conditions, an efficient solution method has been devised to estimate the target position in quasi-closed-form. In particular, by means of an ad-hoc partition of the feasible set, a finite number of candidate optimal solutions has been identified, whose evaluation just rely on the computation of elementary functions and of the roots of specific polynomial equations. For this last task a smart rooting method has been designed capitalizing the structure of the involved equations and the bisection method. Remarkably, the overall target localization process demands a computational complexity proportional to the squared number of receive units. 

The performance of the proposed algorithm has been assessed in terms of RMSE also in comparison with some competitors available in the open literature. For the considered case studies, the new method achieves interesting accuracy gains over the counterparts, especially for weak target returns. Besides, it exhibits performance levels close to the CRLB benchmark further corroborating its effectiveness.

Furthermore, as possible future research avenues, it is definitely of interest to perform an experimental validation of the proposed algorithm on measured data, as well as extend the developed framework to a C-RFM comprising multiple transmitters.

\section{Appendix}\label{appendix}

This Appendix comprises two parts. Part A discusses the regularity of the feasible points for Problem $\cal P$ and then provides the proof of Proposition \ref{Prop1}. Part B deals with the design of computationally efficient techniques to identify candidate optimal solutions.

\subsection{Proof of Proposition \ref{Prop1}}\label{AppendixA}
\begin{lemma}\label{Lemma1}
Any feasible point ${\bar{\bx}}$ to \eqref{eq_prob_comp2} is regular for the optimization Problem ${\cal P}$.
\end{lemma}
\begin{proof}
Let $\bar{\bx}=[\bar{x}_1,\bar{x}_2,\bar{x}_3]^T$ be a feasible point to $\cal P$. Note that $\bar{x}_1 > 0$, i.e., the inequality constraint \eqref{eq:constraint4} is inactive.
In fact, due to the constraint \eqref{eq:constraint1}, $\bar{\bx}=[0,0,0]^T$ cannot be a feasible point.
To study the regularity of $\bar{\bx}$, the following situations should be distinguished:
\begin{enumerate}
\item constraints \eqref{eq:constraint2} and \eqref{eq:constraint3} are simultaneously inactive. The gradient of $\bar\bx^T\bar\bx$ is
$$
\left.\nabla {\bp}^T {\bp}\right|_{{\bp}=\bar{\bx}} = 2 [\bar{x}_1,\bar{x}_2,\bar{x}_3]^T \neq \bzero
$$
implying the regularity of $\bar{\bx}$.

\item $\bar{x}_2 = \gamma_a\bar{x}_1$ (or $\bar{x}_2 = -\gamma_a\bar{x}_1$) and \eqref{eq:constraint3} are inactive. The gradients
$$
\left.\nabla {\bp}^T {\bp}\right|_{{\bp}=\bar{\bx}}= 2 [\bar{x}_1,\gamma_a \bar{x}_1,\bar{x}_3]^T
$$
$$
(\text{or } \left.\nabla {\bp}^T {\bp}\right|_{{\bp}=\bar{\bx}} = 2 [\bar{x}_1,-\gamma_a \bar{x}_1,\bar{x}_3]^T)
$$
and 
$$
\nabla \left(p_2 - \gamma_a p_1\right) = [-\gamma_a,1,0]^T
$$
$$
(\text{or } \nabla \left(-p_2 -\gamma_a p_1\right) = [-\gamma_a,-1,0]^T)
$$
are linearly independent and hence $\bar{\bx}$ is regular.

\item $\bar{x}_3 = \gamma_e \bar{x}_1$ (or $\bar{x}_3 = -\gamma_e \bar{x}_1$) and \eqref{eq:constraint2} are inactive. The gradients
$$
\left.\nabla {\bp}^T {\bp}\right|_{{\bp}=\bar{\bx}}= 2 [\bar{x}_1,\bar{x}_2,\gamma_e \bar{x}_1]^T
$$
$$
(\text{or } \left.\nabla {\bp}^T {\bp}\right|_{{\bp}=\bar{\bx}} = 2 [\bar{x}_1, \bar{x}_2,-\gamma_e\bar{x}_1]^T)
$$
and 
$$
\nabla \left(p_3 - \gamma_e p_1\right) = [-\gamma_e,0,1]^T
$$
$$
(\text{or } \nabla \left(-p_3 -\gamma_e p_1\right) = [-\gamma_e,0,-1]^T)
$$
are linearly independent and hence $\bar{\bx}$ is regular.

\item $\bar{x}_3= \gamma_e \bar{x}_1  $ and $\bar{x}_2=\gamma_a \bar{x}_1$ (or $\bar{x}_2=-\gamma_a \bar{x}_1$). The gradients
$$
\left.\nabla {\bp}^T {\bp}\right|_{{\bp}=\bar{\bx}}= 2 [\bar{x}_1,\gamma_a \bar{x}_1,\gamma_e\bar{x}_1  ]^T
$$
$$
(\text{or }\left.\nabla {\bp}^T {\bp}\right|_{{\bp}=\bar{\bx}} = 2 [\bar{x}_1,-\gamma_a \bar{x}_1,\gamma_e \bar{x}_1  ]^T),
$$

$$
\nabla \left(p_3 - \gamma_e p_1\right) = [-\gamma_e,0,1]^T
$$

and
$$
\nabla \left(p_2 -\gamma_a p_1\right) = [-\gamma_a,1,0]^T
$$
$$
(\text{or } \nabla \left(-p_2 -\gamma_a p_1\right) = [-\gamma_a,-1,0]^T)
$$
are linearly independent implying the regularity of $\bar{\bx}$.
\item $\bar{x}_3= -\gamma_e \bar{x}_1  $ and $\bar{x}_2=\gamma_a \bar{x}_1$ (or $\bar{x}_2=-\gamma_a \bar{x}_1$). The gradients
$$
\left.\nabla {\bp}^T {\bp}\right|_{{\bp}=\bar{\bx}}= 2 [\bar{x}_1,\gamma_a \bar{x}_1,-\gamma_e\bar{x}_1  ]^T
$$
$$
(\text{or }\left.\nabla {\bp}^T {\bp}\right|_{{\bp}=\bar{\bx}} = 2 [\bar{x}_1,-\gamma_a \bar{x}_1,-\gamma_e \bar{x}_1  ]^T),
$$

$$
\nabla \left(-p_3 - \gamma_e p_1\right) = [-\gamma_e,0,-1]^T
$$

and
$$
\nabla \left(p_2 -\gamma_a p_1\right) = [-\gamma_a,1,0]^T
$$
$$
(\text{or } \nabla \left(-p_2 -\gamma_a p_1\right) = [-\gamma_a,-1,0]^T)
$$
are linearly independent implying the regularity of $\bar{\bx}$.

\end{enumerate}

Following the same line of reasoning, it can be also shown that the feasible points of the restricted versions of $\cal P$, obtained considering the different regions of the feasible set, fulfill the regularity condition. 

\end{proof}

\subsubsection{Proof of Proposition \ref{Prop1}}\label{AppendixProp1}
Let us first observe that Weierstrass theorem ensures the existence of a global minimizer to $\cal P$, being the objective function continuous and the constraint set compact.
The basic idea behind the proof is to establish candidate optimal solutions among the feasible points of the problem, which are all regular according to Lemma \ref{Lemma1}. To this end, different regions of the feasible set are explored. 

\begin{enumerate}

\item[a)] Assuming all the inequality constraints inactive, candidate optimal solutions to $\cal P$ can be found among the regular points of

\begin{equation}
{\cal P}_1 \left\{
\begin{array}{ll}
\displaystyle{\min_{\bp}} & { \left|\left| \bH\bp -\bg \right|\right|^2} \\
\text{s.t.} & 
\|\bp\|^2 = \bar{b}_0^2
\end{array}
\right.,
\end{equation}

which satisfy the necessary first-order optimality conditions \cite{Bertsekas16}, as well as the inequality constraints

\begin{equation} \label{constraint_beck4}
\begin{cases}
- \gamma_a {x}_p^\star < {y}_p^\star < \gamma_a  {x}_p^\star\\
- \gamma_e {x}_p^\star < {z}_p^\star <\gamma_e  {x}_p^\star\\
{x}_p^\star > 0
\end{cases}.
\end{equation}

These solutions\footnote{
The solutions in \eqref{sol1} implicitly assume that $(\bH^T\bH + \bar\lambda_h\bI)$ is full-rank. However, almost surely the necessary condition $(\bH^T\bH + \bar\lambda_h\bI)\bp = \bH^T\bg $ when $(\bH^T\bH + \bar\lambda_h\bI)$ is rank deficient does not admit solution, provided that $\bH$ is full-column rank.
} are among the points 
\begin{equation}\label{sol1}
  {\bx}^\star(\bar\lambda_h) = \left({\bH}^T{\bH} + \bar\lambda_h \bI\right)^{-1}{\bH}^T\bg 
\end{equation}
with $\bar\lambda_h,~h \in \bar I_1 \subseteq \{1,\ldots,6\}$, the real-valued roots of the sixth-order equation 
\begin{equation} \label{one_root}
   {\bx}^\star(\bar\lambda)^T {\bx}^\star(\bar\lambda) =\bar{b}_0^2.
\end{equation}

As a consequence, there are at most six candidate optimal points to $\cal P$ for case a).

\item[b)] If $y_p = (-1)^{i+1} \gamma_a x_p, i =1,2$, then  ${\cal P}$  is equivalent to

\begin{equation}
{\cal P}^i_2 \left\{
\begin{array}{ll}
\displaystyle{\min_{\bq}} & { \left|\left| {\bH_i^a}\bq -\bg \right|\right|^2} \\
\text{s.t.} & 
\bq^T \bB^a\bq= \bar{b}_0^2\\
& -q_1\gamma_e \leq q_2 \leq q_1\gamma_e \\
& q_1 \geq 0
\end{array}
\right.,
\end{equation}

where $ \bq = \left[ x_p, z_p \right]^T$,
\begin{align}
{\bH_i^a} = \bH \left[\begin{matrix} 1 & 0 \\
(-1)^{i+1}\gamma_a & 0 \\
0 & 1
\end{matrix}\right], i=1,2, 
\end{align}
and
\begin{equation}
 \bB^a  = \left[\begin{matrix} 1 +\gamma_a^2 & 0 \\0 & 1 \end{matrix}\right]. 
\end{equation}

Assuming $-q_1\gamma_e < q_2 < q_1\gamma_e$ and $q_1 > 0$, candidate optimal solutions to ${{\cal P}_2}$ can be found among the feasible points of

\begin{equation}
{\cal P}^i_3 \left\{
\begin{array}{ll}
\displaystyle{\min_{\bq}} & { \left|\left| {\bH_i^a}\bq -\bg \right|\right|^2} \\
\text{s.t.} & 
\bq^T \bB^a\bq= \bar{b}_0^2
\end{array}
\right.,
\end{equation}

which comply with the necessary optimality conditions and satisfy  $-q_1\gamma_e < q_2 < q_1\gamma_e$ and $q_1 > 0$. These solutions can be obtained from the points\footnote{A situation similar to footnote 5 occurs, i.e., almost surely candidate optimal solutions demand ${\bH_i^a}^T{\bH_i^a} + \beta^i_h \bB^a$ to be full-rank.}

\begin{equation}
    \bq^*(\beta^i_h) = \left( {\bH_i^a}^T{\bH_i^a} + \beta^i_h \bB^a\right)^{-1}{\bH_i^a}^T\bg
\end{equation}

with $\beta^i_h, h\in \bar I_2^i\subseteq \{1,...,4\}$, the real-valued roots to the fourth-order equation 

\begin{equation}
  {\bq^*}^T(\beta^i_h) \bB^a \bq^*(\beta^i_h) = \bar{b}_0^2. 
\end{equation}

As a consequence, there are at most eight candidate optimal points to $\cal P$ for case b) with inequalities strictly satisfied, obtained as $[q_1^*(\beta^i_h), (-1)^{i+1} \gamma_a q_1^*(\beta^i_h), q_2^*(\beta^i_h)]^T, i =1,2, ~h\in I_2^i\subseteq \bar I_2^i$.

\item[c)] If $z_p=(-1)^{i+1}\gamma_e x_p, i=1,2$, and $y_p\neq (-1)^{j+1}\gamma_a x_p, j=1,2$, the same technique as in case b) is used.

    \item[d)] If $y_p=(-1)^{i+1}\gamma_a x_p$ and $z_p =(-1)^{j+1}\gamma_e x_p, (i, j)\in \{1, 2\}^2$,  the candidate solutions are the four points

$~~~~~ {\bx^*_4}_{i,j}=\frac{\bar b_0}{\sqrt{1+ \gamma_a^2 + \gamma_e^2}}\left[ 1, (-1)^{1+i} \gamma_a ,  (-1)^{1+j} \gamma_e  \right]^T$, 

\begin{align}
&~(i, j)\in \{1, 2\}^2. ~~~~~~~~~~~~~~~~~~~~~~~~~~~~~~~~~~ 
\end{align}
\end{enumerate}

\subsection{Efficient Techniques to Identify Candidate Solutions}\label{appendix:sol_det}
To solve the considered 3D localization problem, an efficient procedure is required to identify the real-valued solutions to the equations \eqref{6_order_equation}, \eqref{4_order_equation_1}, and \eqref{4_order_equation_2}. To this end, let us focus on equation \eqref{6_order_equation} which, denoting the eigenvalue decomposition of $\bC=\bH^T\bH$ by $\bU \mbox{diag}([\lambda_1,\lambda_2,\lambda_3]^T)\bU^T$  with $0\leq\lambda_1\leq\lambda_2\leq\lambda_3$ and after some manipulations, can be rewritten as
\begin{equation}\label{6_order_equation_recast}
    \sum_{j=1}^3 \displaystyle{ \frac{|z_j|^2}{(\bar{\lambda} + \lambda_j)^2}}={\bar{b}_0^2},
\end{equation}
where $\bz=\bU^T \by$ and $\by =\bH^T \bg $. Remarkably, since the eigenvalues and eigenvectors of $\bC$ can be computed through elementary functions of the entries of $\bH$, the parameters involved in (\ref{6_order_equation_recast}) are available in closed-form. 
Evidently, solving \eqref{6_order_equation_recast} is tantamount to determining the roots of  
\begin{equation}\label{6_order_function_normalized}
    \bar{f}(\bar \lambda) =  \frac{\bar{z}^2_1}{(\bar\lambda + \lambda_1)^2} + \frac{\bar{z}^2_2}{(\bar\lambda + \lambda_2)^2} + \frac{\bar{z}^2_3}{(\bar\lambda + \lambda_3)^2}-1,
\end{equation}
where $\bar{z}_j=\frac{z_j}{{|\bar{b}_0|}}, j=1,2,3$.
To proceed further, let us observe that $\bar{f}(\bar \lambda)$ is strictly convex within each of the four intervals
$\mathcal{J}_1 = (-\infty,-\lambda_3),$
$\mathcal{J}_2 = (-\lambda_3, -\lambda_2),$
$\mathcal{J}_3 = (-\lambda_2, -\lambda_1),$ and
$\mathcal{J}_4 = (-\lambda_1,+\infty)$,
being the second-order derivative of $\bar{f}(\bar \lambda)$ always positive. Besides, $\bar{f}(\bar\lambda)$ is strictly increasing over $\mathcal{J}_1$ and strictly decreasing over $\mathcal{J}_4$, with $\mbox{lim}_{\bar{\lambda}\rightarrow\mp\infty}\bar{f}(\bar{\lambda})=-1$. Leveraging the above results, it follows that
\begin{itemize}
    \item there exists a unique root of (\ref{6_order_function_normalized}) within  $\mathcal{J}_1$ and another one (still unique) over $\mathcal{J}_4$;
    \item with reference to the intervals $\mathcal{J}_2$ and $\mathcal{J}_3$, the existence of roots depends on the minimum value $v^*_i$ of \eqref{6_order_function_normalized} within $\mathcal{J}_2$ and $\mathcal{J}_3$, respectively. In particular, if $v^*_i> 0$ then \eqref{6_order_function_normalized} does not admit roots belonging to $\mathcal{J}_2$ (or $\mathcal{J}_3$), otherwise there exist two roots if  $v^*_i<0$ and a unique one if $v^*_i=0$.
\end{itemize}
Hence, the unique roots in $\mathcal{J}_1$ and $\mathcal{J}_4$, can be computed via the  bisection algorithm \cite{Analisi_Numerica}. As to the intervals $\mathcal{J}_2$ and $\mathcal{J}_3$, a two-step strategy is now illustrated. At the first stage, the global minimum solution $\bar{\lambda}_i^*$, over $\mathcal{J}_i, i=2,3$,  and the corresponding objective value $v_i^*$ are determined resorting to the bisection method applied over $\mathcal{J}_i, i=2,3$, to the first-order derivative of $\bar{f}(\bar \lambda) $, i.e,
\begin{equation}
\bar{f}'(\bar \lambda)=  -2\left(\frac{\bar{z}^2_1}{(\bar \lambda + \lambda_1)^3} + \frac{\bar{z}^2_2}{(\bar \lambda+ \lambda_2)^3} + \frac{\bar{z}^2_3}{(\bar \lambda + \lambda_3)^3}\right).
\end{equation}
Then, the possible roots of \eqref{6_order_function_normalized} are searched. Specifically, if $v_i^*<0$ two distinct roots exist, $\bar{\lambda}_{i,1} < \bar{\lambda}_{i,2}$ say, which are obtained carrying out the bisection method to the function $\bar{f}(\bar \lambda) $ over the intervals $(- \lambda_{3},\bar{\lambda}_i^*)$ and $(\bar{\lambda}_i^*, - \lambda_2)$ (or $(- \lambda_{2},\bar{\lambda}_i^*)$ and $(\bar{\lambda}_i^*, - \lambda_1)$), respectively. Otherwise, the root is either $\bar{\lambda}_i^*$ or does not exist, if  $v_i^* > 0$.
Before proceeding further two important remarks are now in order:\\
{\bf Remark 1.} Denoting by $\epsilon$ the desired accuracy level for the bisection method,  $\bar{\lambda}_i^*$ and $v_i^*$ may differ from the bisection output $\hat{\bar{\lambda}}_i^*$ and the corresponding objective value $\hat{v}_i^*$ at most by $\epsilon/2$ and $|\bar{f}'(\hat{\bar{\lambda}}_i^*)|\epsilon/2$, respectively. Now, if
$\hat{v}_i^*-|\bar{f}'(\hat{\bar{\lambda}}_i^*)|\epsilon/2>0$ it is guaranteed the absence of roots. Otherwise, even if $\hat{v}_i^*>0$, possible roots may exist, whose $\epsilon$-approximation can be evaluated leveraging $\hat{\bar{\lambda}}_i^*$. Indeed, depending on the sign of $\bar{f}'(\hat{\bar{\lambda}}_i^*)$, the potential roots (if existing) must belong to either $[\hat{\bar{\lambda}}_i^*-\epsilon/2,\hat{\bar{\lambda}}_i^*]$ or $[\hat{\bar{\lambda}}_i^*,\hat{\bar{\lambda}}_i^*+\epsilon/2]$. As a result, either the pair $(\hat{\bar{\lambda}}_i^*-\epsilon/2,\hat{\bar{\lambda}}_i^*)$  or $(\hat{\bar{\lambda}}_i^*,\hat{\bar{\lambda}}_i^*+\epsilon/2)$ can be used to compute candidate optimal solutions with a desired accuracy, which will be automatically discarded, during the screening of the candidates, if the roots do not exist. \\
{\bf Remark 2.} According to the proposed strategy, it is required to execute, in general, three times the bisection method, once at the first stage and twice at the second. However, it is possible to avoid the first stage and determine the two potential roots with at most two bisection cycles, conceiving a bisection-like method: at each iteration, it jointly accounts for the sign of the derivative in correspondence of the two  extremes of the current bisection interval as well as the objective value at the center of the mentioned interval, to update the extremes.

Following the same guideline, the solutions of equations \eqref{4_order_equation_1} {and} \eqref{4_order_equation_2} can be  obtained. {It is also worth observing that} \eqref{4_order_equation_1} and \eqref{4_order_equation_2} could be, in principle, solved in closed-form. However, numerical errors have been experienced {demanding the development} of the aforementioned {numerically robust solution} method.

In the next subsection, details on the bisection initialization are illustrated.
\subsubsection{Bisection Initialization}\label{equation_resolution}
Without loss generality, let us focus on equation \eqref{6_order_equation}. To this end, let us first consider the root search over $\mathcal{J}_1$. Being
\begin{equation}\label{majorant}
   \bar{f}(\bar\lambda) \leq   \frac{\|\bar{\bz}\|^2}{(\bar\lambda + \lambda_3)^2} - 1,\,\,\,\bar\lambda\leq -\lambda_3
\end{equation}
and 
\begin{equation}\label{minorant}
    \bar{f}(\bar\lambda) \geq    \frac{\bar{z}^2_3}{(\bar\lambda + \lambda_3)^2} - 1,\,\,\,\bar\lambda\leq -\lambda_3
\end{equation}
with $\bar{\bz}=[\bar{z}_1,\bar{z}_2,\bar{z}_3]^T$, it follows that the root $\bar{\lambda}_3\in\mathcal{J}_1$ of \eqref{6_order_equation} complies with
$$\bar{\lambda}_3\in[-\lambda_3 - \|\bar{\bz}\|,-\lambda_3 - |\bar{z}_3|].$$
As a consequence, the bisection method can be initialized with the search interval $[-\lambda_3 - \|\bar{\bz}\|,-\lambda_3 - |\bar{z}_3|]$. Leveraging a similar line of reasoning, it stems that  $[-\lambda_1 - \|\bar{\bz}\|,-\lambda_1 - |\bar{z}_1|]$ can be used to initialize the bisection method over the interval $\mathcal{J}_4$.

As to the roots lying within $\mathcal{J}_2$ (analogous reasoning applies for $\mathcal{J}_3$) the initialization at the first stage can be set as $[-\lambda_{3},-(\lambda_{3}+\lambda_{2})/2]$, if $\bar{f}'(-(\lambda_{3}+\lambda_{2})/2)>0$ or $[-(\lambda_{3}+\lambda_{2})/2, -\lambda_{2}]$, if $\bar{f}'(-(\lambda_{3}+\lambda_{2})/2)<0$. The second step, as already said, substantially employs $[-\lambda_{3},\bar{\lambda}_i^*]$ and $[\bar{\lambda}_i^*, -\lambda_{2}]$
to initialize the two bisections.

\bibliographystyle{IEEEtran}
\bibliography{IEEEabrv,ms}

\end{document}